\documentclass[conference]{IEEEtran}
\ifCLASSINFOpdf
\else
\fi
\usepackage{graphicx}
\usepackage{epstopdf}
\usepackage[cmex10]{amsmath}
\usepackage{amssymb}
\usepackage{mdwmath}
\usepackage{mdwtab}

\newtheorem{mydef}{Definition}
\newtheorem{mylemma}{Lemma}

\newtheorem{proof}{Proof}

\newcommand{\qed}{\nobreak \ifvmode \relax \else
      \ifdim\lastskip<1.5em \hskip-\lastskip
      \hskip1.5em plus0em minus0.5em \fi \nobreak
      \vrule height0.75em width0.5em depth0.25em\fi}

\newcommand{\myvspace}{\vspace{1mm}}

\begin{document}
%
\title{Flow of Information in Feed-Forward Deep Neural Networks}

\author{\IEEEauthorblockN{Pejman Khadivi}
\IEEEauthorblockA{Computer Science Department\\
Virginia Tech, Virginia, USA\\
Email: pejman@vt.edu}
\and
\IEEEauthorblockN{Ravi Tandon}
\IEEEauthorblockA{Electrical and Computer Engineering Department\\
University of Arizona,
Arizona, USA\\
Email: tandonr@email.arizona.edu }
\and
\IEEEauthorblockN{Naren Ramakrishnan}
\IEEEauthorblockA{Computer Science Department\\
Virginia Tech, Virginia, USA\\
Email: naren@cs.vt.edu}}


%


\maketitle

\begin{abstract}
Feed-forward deep neural networks have been used extensively in various machine learning applications. Developing  a  precise  understanding  of  the  underling behavior of neural networks is crucial for their efficient deployment. In this paper, we use an information theoretic approach to study the flow of information in a neural network and to determine how entropy of information changes between consecutive layers. Moreover, using the Information Bottleneck principle, we develop a constrained optimization problem that can be used in the training process of a deep neural network. Furthermore, we determine a lower bound for the level of data representation that can be achieved in a deep neural network with an acceptable level of distortion.
\end{abstract}


%
\IEEEpeerreviewmaketitle

\section{Introduction}
With the increasing demand for data analytics, Big Data, and artificial intelligence, efficient machine learning algorithms are required now more than anytime before \cite{chen14}. Deep learning and deep neural networks (DNNs) have been shown to be among the most efficient machine learning paradigms, specifically for supervised learning tasks. Due to their fascinating performance, different deep learning structures have been deployed in various applications in the past decades \cite{chen14,sun14,lenz15}. However, despite of their great performance, more theoretical effort is required to understand the dynamic behavior of DNNs both from learning and design perspectives.

Deep neural networks are considered as multi-layer structures, constructed by simple processing units known as neurons that process the input information to generate a desired output \cite{tom97,nntr14}. These structures have been used previously in a variety of applications, such as dimensionality reduction \cite{hinton06}, face representation \cite{sun14}, robotic grasps detection \cite{lenz15}, and object detection \cite{erhan13}. 

While DNNs have shown their capability in solving machine learning problems, they have been traditionally deployed in a heuristic manner \cite{tishby15,bishop96}. However, to be able to use these structures more efficiently, we need to have a deeper understanding of their underling dynamic behavior \cite{bishop96}.

Mehta and Schwab shown in \cite{mehta14} that deep learning is related to renormalization groups in theoretical physics and provide a mapping between deep learning methods and variational renormalization groups using  Restricted Boltzmann Machines. 
In \cite{tishby15}, Tishby and Zaslavsky proposed a theoretical framework, based on the principle of Information Bottleneck \cite{tishby99}, to analyze the DNNs where, the ultimate goal of deep learning has been formulated as a trade-off between compression and prediction. In \cite{tishby15}, the authors claim that an optimal point exists on the compression-distortion plane that can efficiently address that trade-off. Moreover, they suggest that an Information Bottleneck based learning algorithm may achieve the optimal information representation.

In this paper, we analyze the flow of information in a deep neural network using an information theoretic approach. While different structures have been developed for DNNs, in this paper, we consider the multi-layer feed-forward structure and we assume that the network is used in a supervised setting. We determine an upper bound on the total compression rate in a neural network that can be achieved with an acceptable level of distortion. Furthermore, using the fundamental concepts of Information Bottleneck and based on the approach of Tishby and Zaslavsky in \cite{tishby15}, we develop an optimization problem that can be used in the learning process of DNNs. A case study supports the justifications of the paper. Thus, our contributions and the structure of the paper are as follows:
\begin{itemize}
\item In Section II, we focus on the information flow across any two consecutive layers of a DNN by characterizing the relative change in the entropy across layers and also developing some properties of the same.  
\item In Section III, motivated by the Information Bottleneck principle, we define an optimization problem for training a DNN, in which the goal is to minimize the overall log-loss distortion. Moreover, we prove an upper bound on the total data compression which is achievable in a DNN with an acceptable level of distortion. 
\item In Section IV we modify the optimization problem of Section III to address the practical limitations of neural computation. We first illustrate that how the results of the original optimization model may be unfeasible and then, by adding sufficient constraints to the model we propose a modified optimization problem that can be used in the training process of a DNN. 
\end{itemize}

\begin{figure}[t]
\centering
\includegraphics[trim=0.2cm 0.4cm 0.4cm 0.1cm, clip=true,width=0.4\textwidth]{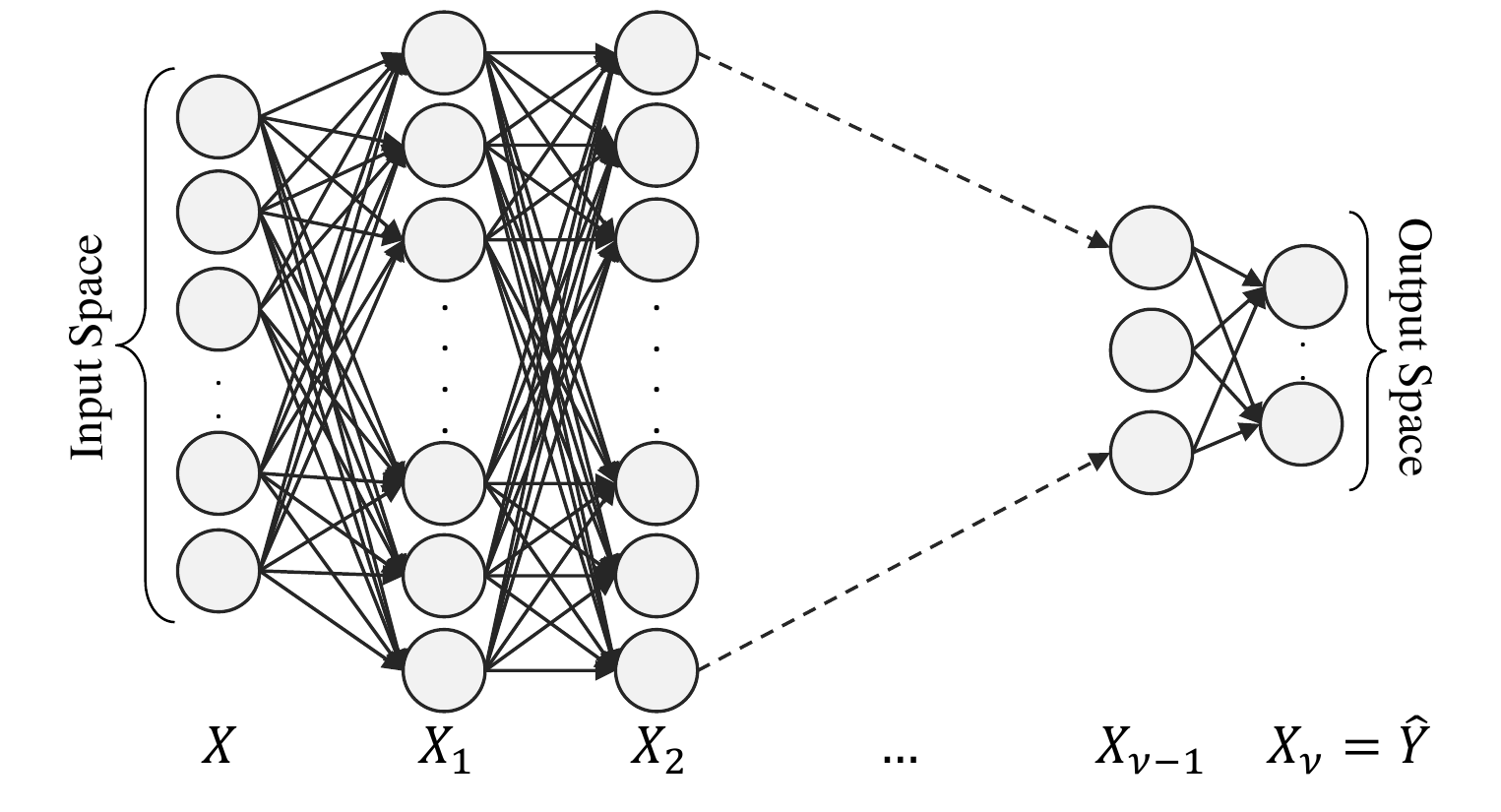}
\caption{General structure of a feed-forward deep neural network.}
\label{fig:1}
\vspace{-5mm}
\end{figure}

\section{Flow of Entropy in Deep Neural Networks}
\label{sec:flow}
The typical structure of a feed-forward DNN is illustrated in Fig \ref{fig:1}. In this figure, input layer is represented by $X$ and the representation of data in the $i^{th}$ hidden layer is shown by $X_i$. We assume that the network has $\nu$ layers and the output layer, i.e. $X_{\nu}$, should estimate a desired output, $Y$. Each layer is constructed by multiple neurons that process the information in parallel and the output of the $j^{th}$ neuron in layer $n$ is calculated as follows:
\begin{equation}
\label{eq:neuron}
x_{n,j} = g^n_j(x_{n-1}) = f^n\left (\sum_{i=1}^{m_{n-1}}{w^{n}_{i,j} x_{n-1,i}}+b^{n}_{j}  \right )
\end{equation}
where $m_{n-1}$ is the number of neurons in layer $n-1$ and $w^{n}_{i,j}$ is a weight that connects the output of the $i^{th}$ neuron in layer $n-1$ to the input of the $j^{th}$ neuron in layer $n$. Also, $b^{n}_{j}$ is the bias of the $j^{th}$ neuron in layer $n$ and $f^n(\cdot)$ is the output function of the neurons in this layer. To illustrate (\ref{eq:neuron}) in vector notation we say:
\begin{equation}
\mathbf{x}_{n} = G^n(\mathbf{x}_{n-1})
\end{equation}
where $\mathbf{x}_{n}$ is a combination of neuron outputs in layer $n$, i.e. $\mathbf{x^n} = [x_{n,1}, \cdots, x_{n,m_n}]$. The total number of possible output combinations in the $n^{th}$ layer is illustrated by $k_n$ which depends on the output functions and the input space. As an example, with binary output functions, $k_n = 2^{m_n}$. At the $n^{th}$ layer, $\mathcal{X}_n = \{\mathbf{x^n_1},\cdots,\mathbf{x^n_{k_n}}\}$ is the set of all possible output combinations. It should be noted that a single neuron has a very limited processing capability. As a matter of fact, an individual neuron can only implement a hyper-plane in its input space and hence, not all the mappings from its input to its output are feasible. This limitation is one of the main reasons that neural networks are constructed by multiple layers of interacting neurons. 

The neural computation that occurs in each layer performs a mapping between the outputs of the consecutive layers.
In other words, various output combinations in layer $n-1$ are mapped to certain output combinations in layer $n$. Depending on the weights, bias, and the output function in layer $n$, there are two possibilities:
\begin{itemize}
\item Each unique combination of neuron outputs in layer $n-1$ is uniquely mapped to a combination of neuron outputs in layer $n$. In other words:
$$
\forall \mathbf{x_1},\mathbf{x_2}\in \mathcal{X}_{n-1} \therefore G^n(\mathbf{x_1})\neq G^n(\mathbf{x_2})
$$
In this case, regardless of the number of neurons in each layer, we have $k_n=k_{n-1}$.

\item Multiple (at least two) combinations of neuron outputs in layer $n-1$ are mapped to a single combination of neuron outputs in layer $n$. In other words:
$$
\exists \mathbf{x_1},\mathbf{x_2}\in \mathcal{X}_{n-1} \therefore G^n(\mathbf{x_1}) = G^n(\mathbf{x_1})
$$
and in this case, we have $k_{n}<k_{n-1}$.
\end{itemize}
It worth mentioning that the mapping in each layer of DNN is also a partitioning of the layer's input space and hence, what a DNN does is multiple consecutive of partitioning and mapping processes with goal of estimating of a desired output.

\myvspace
\begin{mydef}
In the $n^{th}$ layer of a feed-forward DNN, \textit{layer partition} is the partitioning of the $n^{th}$ layer's input space which occurs due to the neural computations performed in layer $n$. We illustrate this partitioning by $\mathcal{S}_n = \{S^n_1,\cdots,S^n_{k_n}\}$ where, $S^n_j$ is the set of all the output combinations in layer $n-1$ that are mapped to the $j^{th}$ output combination in layer $n$, i.e. $\mathbf{x^{n}_{j}}$. In other words,
$$
S^n_j = \{\,\mathbf{x^{n-1}_i} \in \mathcal{X}_{n-1} \,\,|\quad \mathbf{x^{n}_{j}}=G^n(\mathbf{x^{n-1}_i}) \, \}
$$
Note that we have:
$$
i \neq j \Rightarrow S^n_i \cap S^n_j = \emptyset \quad \text{and} \quad
\bigcup_{j=1}^{k_n}S^n_j = \mathcal{X}_{n-1}.
$$
\end{mydef}

Let us assume that $P_{n}(\mathbf{x^n_j})$ is the probability of $\mathbf{x^n_j}$. Then, it can be observed that
\begin{equation}
P_n(\mathbf{x^n_j}) = \sum_{\mathbf{x'}\in S^n_j}{P_{n-1}(\mathbf{x'})}
\end{equation}
Furthermore, considering this fact that $\mathcal{S}_n$ is a partitioning of $\mathcal{X}_{n-1}$, one can easily show that $P_n(\mathbf{x^n_j})$ is the probability of partition $S^n_j$. It can be shown that 
\begin{equation}
\mathbf{\Pi}^n_j = \left\{\forall \mathbf{x}\in S^n_j \,|\,  \frac{P_{n-1}(\mathbf{x})}{\sum_{\mathbf{x'}\in S^n_j}{P_{n-1}(\mathbf{x'})}} \right\}
\end{equation}
is the probability distribution of all the combinations in $S^n_j$. 

\myvspace
\begin{mydef}\label{df:entropy}
In a feed-forward DNN, the \textit{entropy of layer} $n$, $H(X_{n})$, is the entropy of the neuron outputs at this layer and we have:
\begin{equation*}
H(X_{n}) = -\sum_{i=1}^{k_{n}}{P_{n}(\mathbf{x^{n}_i}) \log{P_{n}(\mathbf{x^{n}_i})}} .
\end{equation*}
\end{mydef}

\myvspace
\begin{mydef}\label{df:entropy2}
The \textit{entropy of partition} $S^n_j$, $H(S^n_j)$, is the entropy of the output combinations that belong to $S^n_j$. In other words:
$$
H(S^n_j) = - \sum_{\mathbf{x}\in S^n_j} \frac{P_{n-1}(\mathbf{x})}{\sum_{\mathbf{x'}\in S^n_j}{P_{n-1}(\mathbf{x'})}} \log \frac{P_{n-1}(\mathbf{x})}{\sum_{\mathbf{x'}\in S^n_j}{P_{n-1}(\mathbf{x'})}}
$$
\end{mydef}

In the following lemma we show how entropy of information changes in a feed-forward neural network.

\myvspace
\begin{mylemma}\label{lm:1}
In a feed-forward DNN, the entropy of information that flows from layer $n-1$ to layer $n$ is decreased by the expected entropy of the partitions on the possible output combinations in layer $n-1$. The amount of this reduction is shown by $\Delta_n$ and we have
\begin{equation}
\label{eq:delta}
\Delta_n = H(X_{n-1})-H(X_n) = \sum_{j=1}^{k_{n}} P_n(\mathbf{x^n_j}) H(S^n_j).
\end{equation}
\end{mylemma}
\begin{proof}
Let us assume that the information has been processed up to layer $X_{n-1}$. Using Definition \ref{df:entropy}, the entropy of layer $n-1$ is
\begin{equation}
\label{eq:ent1}
H(X_{n-1}) = -\sum_{i=1}^{k_{n-1}}{P_{n-1}(\mathbf{x^{n-1}_i}) \log{P_{n-1}(\mathbf{x^{n-1}_i})}} 
\end{equation}
To determine $H(X_n)$, we can say
\begin{align}
\label{eq:ent2}
\nonumber
H(X_n) = & - \left \{\sum_{i=1}^{k_{n-1}}{P_{n-1}(\mathbf{x^{n-1}_i}) \log{P_{n-1}(\mathbf{x^{n-1}_i})}} \right.\\ \nonumber
& +\sum_{j=1}^{k_n}\left [ \left (\sum_{\mathbf{x}\in S^n_j}{P_{n-1}(\mathbf{x})}  \right )\log{\left(\sum_{\mathbf{x}\in S^n_j}{P_{n-1}(\mathbf{x}})\right)} \right.\\ 
& -\left.\left. \left(\sum_{\mathbf{x}\in S^n_j}{P_{n-1}(\mathbf{x}) \log P_{n-1}(\mathbf{x})}\right) \right ]  \right \} 
\end{align}
In other words, we started from the entropy of $X_{n-1}$ and substituted the individual terms of all the output combinations that belong to $S^n_j$ with their new equivalent term, i.e. $\left(\sum_{\mathbf{x}\in S^n_j}{P_{n-1}(\mathbf{x})}\right)\log{\left(\sum_{\mathbf{x}\in S^n_j}{P_{n-1}(\mathbf{x}})\right)}$. On the other hand, we have:
\begin{align*}
& \left (\sum_{\mathbf{x}\in S^n_j}{P_{n-1}(\mathbf{x})}  \right )\log{\left(\sum_{\mathbf{x}\in S^n_j}{P_{n-1}(\mathbf{x}})\right)} \\
& - \left(\sum_{\mathbf{x}\in S^n_j}{P_{n-1}(\mathbf{x}) \log P_{n-1}(\mathbf{x})}\right) \\
& = - \sum_{\mathbf{x}\in S^n_j} P_{n-1}(\mathbf{x}) \log \frac{P_{n-1}(\mathbf{x})}{\sum_{\mathbf{x'}\in S^n_j}{P_{n-1}(\mathbf{x'})}}
\end{align*}
then, considering (\ref{eq:ent1}), we can rewrite (\ref{eq:ent2}) as follows:
\begin{align}
\nonumber
H(X_n) &= H(X_{n-1})\\
& + \sum_{j=1}^{k_n}\sum_{\mathbf{x}\in S^n_j} P_{n-1}(\mathbf{x}) \log \frac{P_{n-1}(\mathbf{x})}{\sum_{\mathbf{x'}\in S^n_j}{P_{n-1}(\mathbf{x'})}}
\end{align}
Equivalently, we have
\begin{align}
\label{eq:ent3}
\nonumber
& H(X_n) = H(X_{n-1}) + \sum_{j=1}^{k_n} \left(\sum_{\mathbf{x'}\in S^n_j}{P_{n-1}(\mathbf{x'})} \right)\\
& \left(\sum_{\mathbf{x}\in S^n_j} \frac{P_{n-1}(\mathbf{x})}{\sum_{\mathbf{x'}\in S^n_j}{P_{n-1}(\mathbf{x'})}} \log \frac{P_{n-1}(\mathbf{x})}{\sum_{\mathbf{x'}\in S^n_j}{P_{n-1}(\mathbf{x'})}}\right)
\end{align}
Then, based on Definition \ref{df:entropy2} we can say:
\begin{equation}
H(X_n) = H(X_{n-1}) - \sum_{j=1}^{k_{n}} P_n(\mathbf{x^n_j}) H(S^n_j)
\end{equation}
and from here, we can observe that $H(X_n) \leq H(X_{n-1})$. Furthermore, the difference between the entropy in layers $n-1$ and $n$ is $\sum_{j=1}^{k_{n}} P_n(\mathbf{x^n_j}) H(S^n_j)$ which is the expected entropy of the partitions on $\mathcal{X}_{n-1}$. This proves the lemma.\qed
\myvspace
\end{proof}

The following lemma proves a similar result for the flow of conditional entropy in a deep neural network.

\myvspace
\begin{mylemma}\label{lm:2}
In a feed-forward DNN, the conditional entropy of each layer, conditioned on the desired outputs, $Y$, i.e. $H(X_n|Y)$, is a non-increasing function of $n$ and we have:
\begin{align}
\nonumber
\Delta_n' &= H(X_{n-1}|Y)-H(X_n|Y)\\
& = \sum_{y\in \mathcal{Y}}{\sum_{j=1}^{k_n}{P_Y(y)P_n(\mathbf{x^n_j}|Y=y) H(S^n_j|y)}}
\end{align}
where
$$
P_n(\mathbf{x^n_j}|Y=y) = \sum_{\mathbf{x}\in S^n_j} P_{n-1}(\mathbf{x}|Y=y)
$$
and
$$
H(S^n_j|y) = -\sum_{\mathbf{x}\in S^n_j} \frac{P_{n-1}(\mathbf{x}|Y=y)}{P_n(\mathbf{x^n_j}|Y=y)}
\log \frac{P_{n-1}(\mathbf{x}|Y=y)}{P_n(\mathbf{x^n_j}|Y=y)}.
$$
\end{mylemma}
\begin{proof}
The proof of Lemma 2 follows on similar lines as Lemma 1 by conditioning on the random variable Y and is therefore omitted.\qed
\myvspace
\end{proof}

\section{Optimal Mapping and Information Bottleneck}
\label{sec:ibopt}
Extraction of relevant information from an input data with respect to a desired output is one of the central issues in supervised learning \cite{dietterich2002}. In information theoretic terms, assuming that $X$ and $Y$ are the input and output random variables, $I(X;Y)$ is the relevant information between $X$ and $Y$. In order to improve the efficiency of a machine learning task, we generally need to extract the minimal sufficient statistics of $X$ with respect to $Y$. Hence, as the Information Bottleneck principle \cite{tishby99} indicates, we need to find a maximally compressed representation of $X$ by extracting the relevant information with respect to $Y$ \cite{tishby15}. In this section, we follow the approach of \cite{tishby15} to define an optimization problem that can be used in the learning process of a DNN. We also prove an upper bound on the achievable compression rate of the input in DNNs.

Consider a DNN with $\nu$ layers. Let us assume that $\hat{X} = [X_1, \cdots, X_{\nu}]$ denotes the output of these layers. To find the maximally compressed representation of $X$, we need to minimize the mutual information between the input, i.e. $X$, and the representation of data by the DNN, $\hat{X}$. This can be formulated as minimizing the mutual information between the input and the representation of data in each layer, i.e. $X_i, i=1,\cdots,\nu$ which can be modeled as
\begin{equation}
\min \sum_{i=1}^{\nu}{I(X;X_i)}.
\end{equation}

To measure the fidelity of the training process of a feed-forward DNN, we focus on the Logarithmic-loss distortion function. Let $p(y|x)$ denotes the conditional distribution of $Y$ given $X$ (i.e., the original input data) and similarly, let $p(y|\hat{x})$ denotes the conditional distribution of $Y$ given $\hat{X}$ (i.e., given the outputs of all the layers). Then, one measure of fidelity is the KL-divergence between these distributions:
\begin{equation}
d_{IB}(x,\hat{x}) = \sum_y P(y|x) \log \frac{P(y|x)}{P(y|\hat{x})}
\end{equation} 
and taking the expectation of $d_{IB}(x,\hat{x})$ we get:
\begin{equation}
D_{IB} = E \left[d_{IB}(x,\hat{x}) \right] = I(X;Y|\hat{X})
\end{equation}
Using the Markov chain properties of the feed-forward network, one can show that
\begin{equation}
I(X;Y|\hat{X}) = I(X;Y\,|\,\,[X_1,\cdots,X_{\nu}]) = I(X;Y|X_{\nu})
\end{equation}
and hence, the overall distortion of the DNN will be $D_{IB} = I(X;Y|X_{\nu})$. 
Therefore, using the Information Bottleneck principle, the training criteria for a feed-forward DNN can be formulated as follows
\begin{align}
\label{eq:opt1}
\nonumber
& \quad\quad \min \sum_{i=1}^{\nu}{I(X;X_i)}\\ 
& s.t.\quad I(X;Y|X_{\nu}) \leq \epsilon
\end{align}

As we have mentioned in Section \ref{sec:flow}, in the $i^{th}$ layer, each input combination is mapped to a specific output combination. Therefore, it can be easily shown that $H(X_i|X)=0$ and hence, $I(X;X_i) = H(X_i)$. Then, using Lemma \ref{lm:1} we have:
\begin{equation}
\label{eq:mi1}
I(X;X_i) = H(X) - \sum_{j=1}^{i}\Delta_j.
\end{equation}
Note that in the above equation, $H(X)$ is constant, i.e. it does not depend on the neural network settings.

Regarding the constraint of (\ref{eq:opt1}), it can be shown that
\begin{equation}
\label{eq:mic}
I(X;Y|X_{\nu}) = I(X;Y)-I(X_{\nu};Y)
\end{equation}
On the other hand, we know that $I(X_{\nu};Y) = H(X_{\nu})-H(X_{\nu}|Y)$ and using Lemma \ref{lm:1} and Lemma \ref{lm:2}, we have
\begin{align*}
\left\{\begin{matrix}
H(X_{\nu}) = H(X) - \sum_{i=1}^{\nu}\Delta_i\\ 
H(X_{\nu}|Y) = H(X|Y) - \sum_{i=1}^{\nu}\Delta_i'
\end{matrix}\right.
\end{align*}
which results in the following equation:
\begin{equation}
\label{eq:mi2}
I(X;Y|X_{\nu}) = \sum_{i=1}^{\nu}\left(\Delta_i - \Delta_i'\right).
\end{equation}

Using (\ref{eq:mi1}) and (\ref{eq:mi2}) and by minor manipulation, the optimization problem of (\ref{eq:opt1}) can be rewritten as follows:
\begin{align}
\label{eq:opt3}
\nonumber
& \quad\quad \max \sum_{i=0}^{\nu-1}(\nu-i)\Delta_{i+1}\\ 
& s.t.
\quad \sum_{i=1}^{\nu}\left(\Delta_i - \Delta_i'\right) \leq \epsilon
\end{align}
This is a convex optimization problem and due to its complexity, it is generally difficult to find an analytic solution for that. However, numerical solutions (such as algorithms based on Blahut-Arimoto \cite{blahut72}) may be deployed here to solve (\ref{eq:opt3}).

In the optimization problem of (\ref{eq:opt3}), $\Delta_i$ is the amount of entropy reduction in layer $i$ which can be interpreted as the amount of data compression that has been occurred at this layer. Moreover, we can observe that the total data compression (i.e. reduction in entropy) that occurs in DNN is $\sum \Delta_i$ and is defined in the following definition:

\myvspace
\begin{mydef}\label{df:comp}
The \textit{total compression} in a feed forward DNN with $\nu$ layers is illustrated by $\mathsf{C}_{\nu}$ and we have:
\begin{equation*}
\mathsf{C}_{\nu} = \sum_{i=1}^{\nu} \Delta_i
\end{equation*}
\myvspace
\end{mydef}

The following lemma shows an upper bound on $\mathsf{C}_{\nu}$:

\myvspace
\begin{mylemma} \label{lm:3}
In a multilayer neural network with input $X$, output layer $X_{\nu}$, and the desired output $Y$, the maximum possible entropy reduction from the input space to the output space that satisfies the distortion constraint is $H(X|Y) - H(X_{\nu}|Y)$.
\end{mylemma}
\begin{proof}
From the constraint of (\ref{eq:opt3}) we have:
\begin{equation}
\label{eq:ub1}
\sum_{i=1}^{\nu}\Delta_i \leq \epsilon+\sum_{i=1}^{\nu}\Delta_i'
\end{equation}
However, we know that $\Delta_i' = H(X_{i-1}|Y)-H(X_{i}|Y)$. Hence, we have
\begin{align*}
& \sum_{i=1}^{\nu}\Delta_i' = H(X|Y)-H(X_1|Y)+ H(X_1|Y)-H(X_2|Y)  \\
& + H(X_2|Y) - \cdots - H(X_{\nu}|Y) = H(X|Y) - H(X_{\nu}|Y)
\end{align*}
Therefore, using Definition \ref{df:comp} and (\ref{eq:ub1}) and when $\epsilon \rightarrow 0$ we have:
\begin{equation}
\label{eq:ub2}
\mathsf{C}_{\nu} \leq H(X|Y) - H(X_{\nu}|Y)
\end{equation}
This proves the lemma.\qed
\myvspace
\end{proof}

The first consequence of Lemma \ref{lm:3} is that, regardless of the number of layers, $\mathsf{C}_{\nu}$ cannot be greater than $H(X|Y)$. In fact, considering Lemma \ref{lm:2}, $H(X_i|Y)$ is a non-increasing function of $i$, and hence we have:
\begin{equation}
\label{eq:ub3}
\mathsf{C}_{\nu} \leq H(X|Y) 
\end{equation}
However, it should be noted that higher number of layers may result in a more compressed representation of data. In other words, based on Lemma \ref{lm:2}, for $\nu_1 > \nu_2$ we have $H(X_{\nu_1}|Y) \geq H(X_{\nu_2}|Y)$ and hence, $\mathsf{C}_{\nu_1} \leq \mathsf{C}_{\nu_2}$. In the next section, we indicate that due to the structural limitations of an artificial neuron, not all the mappings determined by (\ref{eq:opt3}) can be implemented using one layer. Therefore, in addition to have a more compressed representation of information,  in a neural network multiple layers may be required to achieve feasible mappings from the input space to the output space.

\begin{figure}
\centering
\begin{tabular}{cc}
\includegraphics[trim=0.1cm 0.1cm 0.1cm 0.1cm, clip=true,width=0.1\textwidth]{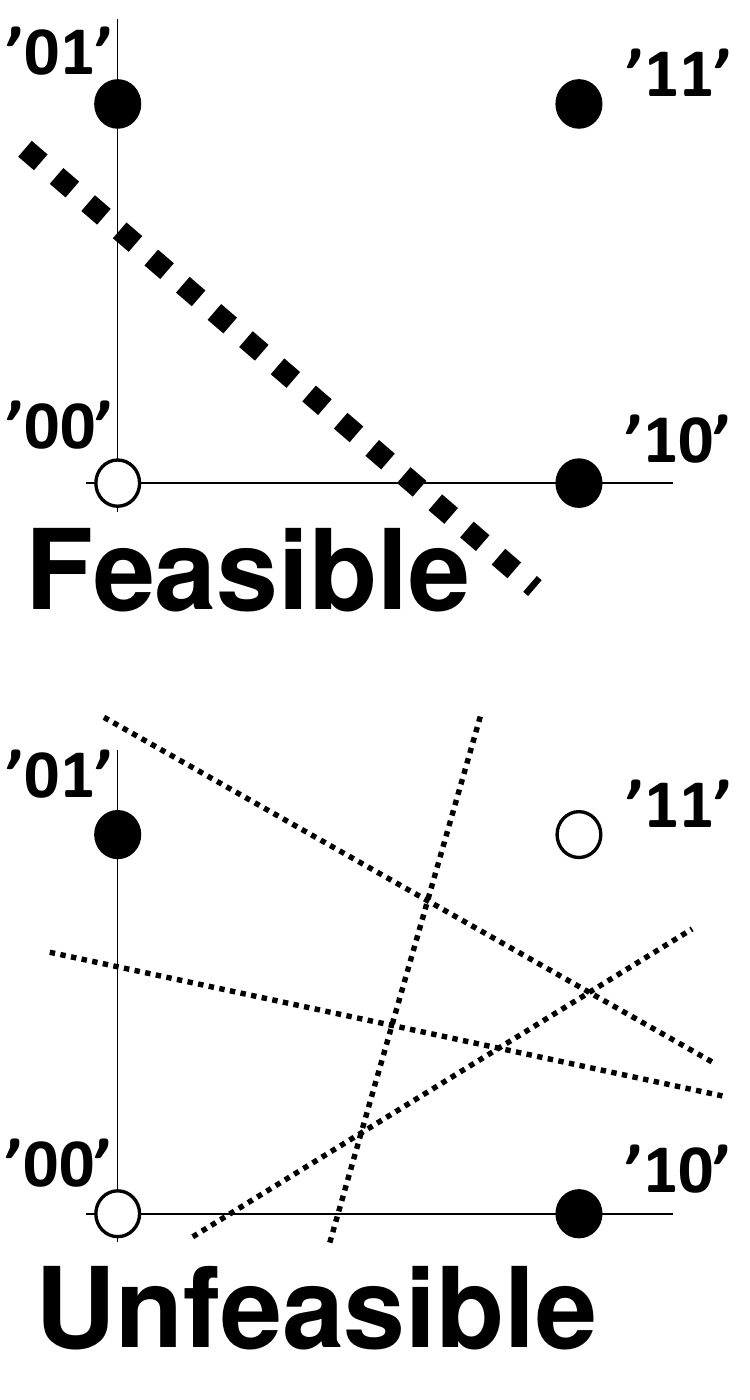} &
\includegraphics[trim=0.1cm 0.1cm 0.1cm 0.1cm, clip=true,width=0.26\textwidth]{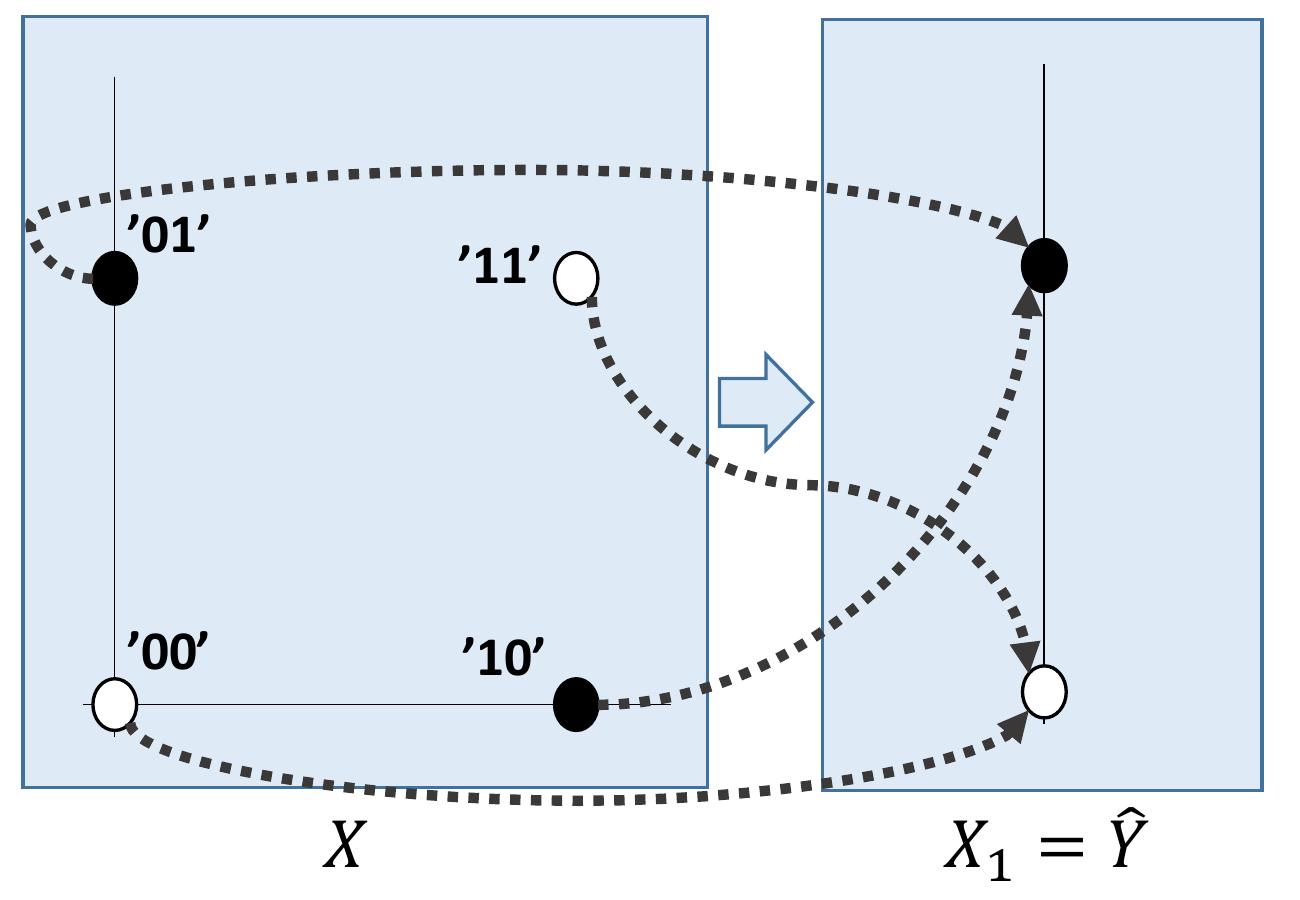}\\
(a) & (b) 
\end{tabular}
\vspace{-1mm}
\caption{(a) Examples of feasible and unfeasible mappings. Black and white circles are mapped to outputs '1' and '0' respectively. (b) Unfeasible mappings resulted from the optimization problem of (\ref{eq:opt3}) with a single neuron.}
\label{fig:unfeasible}
\vspace{-15pt}
\end{figure}

\section{Feasible Optimal Mappings}
\label{sec:feaopt}
While the optimization problem of (\ref{eq:opt3}) can be used to find the optimal mappings between consecutive layers of a neural network, it may result in unfeasible solutions. As a matter of fact, a single neuron implements a single hyperplane in the input space and hence, only linearly separable classification problems may be solved with a single neuron. As an example, in binary input/output space (i.e. when the inputs and the output of the neuron are binary), an XOR function cannot be implemented with a single neuron. Examples of feasible and unfeasible mappings with a single neuron are illustrated in Fig. \ref{fig:unfeasible}(a). In this figure, black and white circles are mapped to outputs '1' and '0', respectively. However, a single neuron can only divide the space into two parts and hence, the top mapping (i.e. boolean OR function) can be implemented by a single neuron while the bottom mapping (i.e. boolean XOR function) is not implementable.
The unfeasible mappings which are resulted from the optimization problem of (\ref{eq:opt3}) for a boolean XOR function are illustrated in Fig. \ref{fig:unfeasible}(b). Therefore, we need to add more constraints to the above optimization problems to exclude the unfeasible mappings.

Using (\ref{eq:delta}) and the definition of $H(S_j^n)$ we can show that
\begin{equation}
\label{eq:newdelta}
\Delta_n = \sum_{j=1}^{k_n}\sum_{\mathbf{x}\in S^n_j} P_{n-1}(\mathbf{x}) \log \frac{P_n(\mathbf{x_j^n})}{P_{n-1}(\mathbf{x})}
\end{equation}
Let us define the following parameter:
\begin{equation}
\theta_{ij}^n=\left \{ \begin{matrix}
1 \quad , \quad \mathbf{x_i^{n-1}}\in \mathbf{S}_j^n\\ 
0 \quad , \quad  Otherwise
\end{matrix} \right.
\end{equation}
where, $\sum_{j=1}^{k_n}\theta_{ij}=1$. Moreover, in vector notations, $\Theta_n$ is a $k_{n-1}\times k_n$ matrix such that $\Theta_n[i,j]=\theta_{ij}^n$.
Then, (\ref{eq:newdelta}) can be written as follows
\begin{equation}
\label{eq:newdelta2}
\Delta_n = \sum_{j=1}^{k_n}\sum_{i=1}^{k_{n-1}} \theta_{ij}^n P_{n-1}(\mathbf{x_i^{n-1}}	) \log \frac{P_n(\mathbf{x_j^n})}{P_{n-1}(\mathbf{x_i^{n-1}})}
\end{equation}
Furthermore, using a similar notation, it can be shown that 
\begin{align}
\label{eq:newdeltap2}
\nonumber
\Delta_n'= &\sum_{y\in \mathcal{Y}}\sum_{j=1}^{k_n}\sum_{i=1}^{k_{n-1}}\theta_{ij}P_Y(y) P_{n-1}(\mathbf{x_i^{n-1}}|Y=y)\\
&\log \frac{P_{n}(\mathbf{\mathbf{x}_j^n}|Y=y)}{P_{n-1}(\mathbf{x_i^{n-1}}|Y=y)}
\end{align}

As we mentioned before, not all the mappings between the inputs and outputs of a neuron are feasible. Unfeasible mappings depend on the structure of the network, number of neurons in each layer, and the corresponding output functions in each layer. Let us assume that at the $n^{th}$ layer, $\mathbf{\tilde{\Theta}_n}=\{\tilde{\Theta}_{1,n},\cdots,\tilde{\Theta}_{{\kappa_n},n}\}$ is the set of forbidden mappings. Then, the optimization problem of (\ref{eq:opt3}) can be modified as follows:
\begin{align}
\label{eq:opt5}
\nonumber
& \max \sum_{n=0}^{\nu-1} \sum_{j=1}^{k_{n+1}}\sum_{i=1}^{k_{n}}(\nu-n)  \theta_{ij}^{n+1} P_{n}(\mathbf{x_i^{n}}) \log \frac{P_{n+1}(\mathbf{x_j^{n+1}})}{P_{n}(\mathbf{x_i^{n}})}\\ \nonumber
& s.t. \\ \nonumber
&  \sum_{n=1}^{\nu}\sum_{y\in \mathcal{Y}}\sum_{j=1}^{k_n}\sum_{i=1}^{k_{n-1}}
\theta_{ij}^n P_Y(y) \left(P_{n-1}(\mathbf{x_i^{n-1}}) \log \frac{P_n(\mathbf{x_j^n})}{P_{n-1}(\mathbf{x_i^{n-1}})} \right.\\ \nonumber
& \left. \quad \quad \quad - P_{n-1}(\mathbf{x_i^{n-1}}|Y=y) \log \frac{P_{n}(\mathbf{\mathbf{x}_j^n}|Y=y)}{P_{n-1}(\mathbf{x_i^{n-1}}|Y=y)}\right) \leq \epsilon \\ \nonumber
& \sum_{j=1}^{k_n} \theta_{ij}^n = 1 \quad , \quad  n=1,\cdots,\nu \quad, \quad i=1,\cdots,k_{n-1} \\
& \Theta_n \neq \tilde{\Theta}_{i,n} \quad , \quad  n=1,\cdots,\nu \quad, \quad i=1,\cdots,\kappa_n
\end{align}
where the solution to (\ref{eq:opt5}) is the set of $\theta_{ij}$'s. Note that the last statement is used to exclude the forbidden mappings from the set of solutions. The optimization problem of (\ref{eq:opt5}) finds the optimal mappings between any two consecutive layers in a feed-forward DNN. These mappings can then be implemented by proper selection of neuron weights.

\section{Case Study: Boolean Functions}
In this section, we perform a case study to observe how the proposed optimization problem of (\ref{eq:opt5}) may be used to determine optimal mappings between consecutive layers in a DNN. For this study, we try to implement basic boolean functions and  show how entropy changes from the input to the output layer. In this set of experiments we use AND, OR, and XOR functions with two and three inputs and we assume that $\nu \in \{1,2,3\}$. Results are illustrated in Table \ref{tbl:1}. It is clear from these results that feasible mappings cannot be determined for two and three input XOR functions when $\nu=1$.  As we have mentioned before, this is due to the processing limitations of a single neuron. However, for AND and OR functions even one single neuron was able to implement the function. Figure \ref{fig:2} shows an example set of mappings between consecutive layers for XOR function using a two-layer neural network.

Table \ref{tbl:1} also illustrates the achievable compression rate, i.e. $\mathsf{C}_{\nu}$, and its corresponding upper bound, i.e. $H(X|Y)$. As we can observe, in the cases that the function was implementable using the neural network, $\mathsf{C}_{\nu}$ is equal to $H(X|Y)$, which means that for these functions we have been able to achieve the minimum representation of data at the output layer. As we proved in Lemma \ref{lm:3}, we observe that the maximum achievable level of data compression in a feed-forward DNN is $H(X|Y)$. Moreover, results indicate that the main reason to add an extra layer to a DNN is to achieve feasible mappings. However, as Lemma \ref{lm:3} shows, extra layers may lead to a more compressed representation of the input data.

\begin{table}
\caption{Case study results for boolean functions.} 
\begin{center}
\scalebox{0.8}{
  \begin{tabular}{ l | c | c | c | c | c | c | }
    \hline
    \hline
    Function & $\nu$ & Neurons  & Solution  & Optimization  & $\mathsf{C}_{\nu}$ & $H(X|Y)$ \\
    (Inputs) & & per Layer & Exists? & Function & & \\ \hline
    AND (2) & 1 & [1] & Yes & 1.189 & 1.189 & 1.189  \\ \hline
    OR (2)$^*$ & 1 & [1] & Yes & 0.888 & 0.888 & 0.888  \\ \hline
    XOR (2) & 1 & [1] & No & - & - & -  \\ \hline
    AND (2) & 2 & [2 1] & Yes & 2.377 & 1.189 & 1.189  \\ \hline
    XOR (2) & 2 & [2 1] & Yes & 1.500 & 1.000 & 1.000  \\ \hline
    AND (2) & 3 & [2 2 1] & Yes & 3.566 & 1.189 & 1.189  \\ \hline
    XOR (2) & 3 & [2 2 1] & Yes & 2.500 & 1.000 & 1.000  \\ \hline
    AND (3) & 1 & [1] & Yes & 2.456 & 2.456 & 2.456  \\ \hline
    XOR (3) & 1 & [1] & No & - & - & -  \\ \hline
    \hline
  \end{tabular}
  }
\end{center}
 \scalebox{0.65}{
  (*) Distribution is not uniform: Probability of $00$ is 0.7 and others are 0.1.}
  \label{tbl:1}
\end{table}

\begin{figure}[t]
\centering
\includegraphics[trim=0.1cm 0.1cm 0.1cm 0.1cm, clip=true,width=0.38\textwidth]{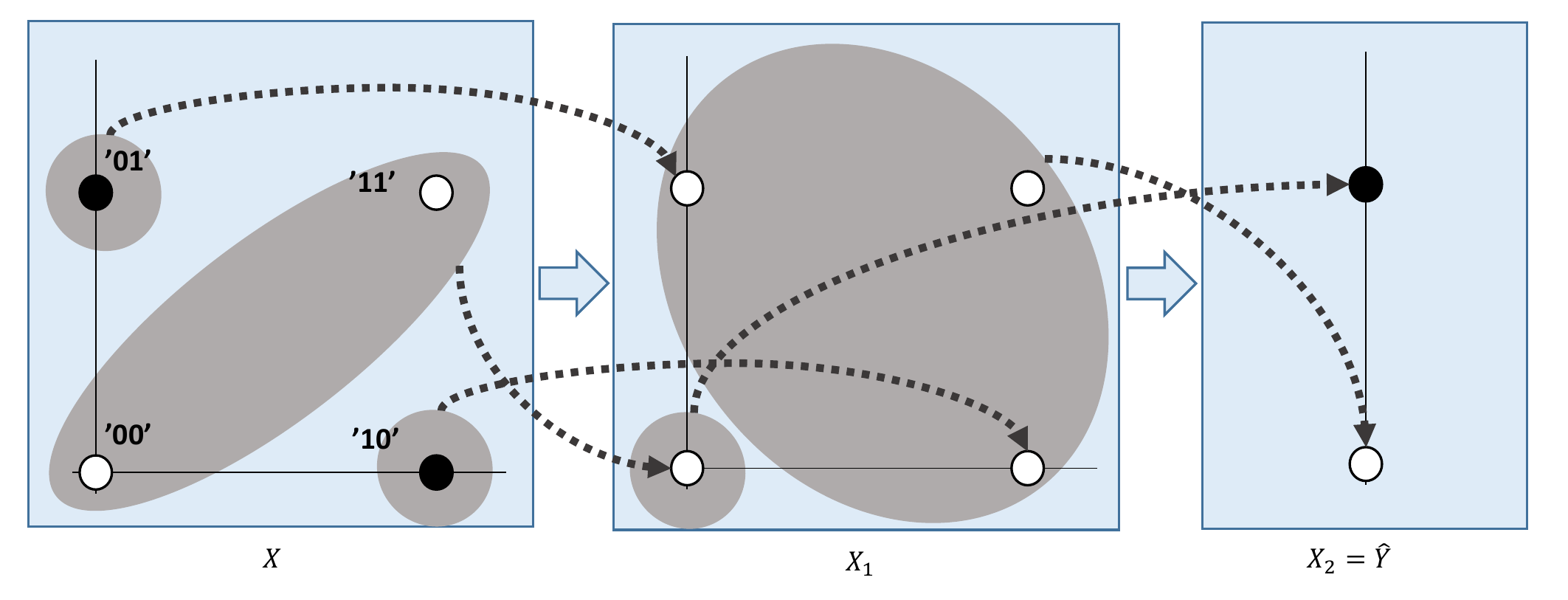}
\vspace{-2mm}
\caption{An example sequence of mappings for XOR.}
\label{fig:2}
\vspace{-3mm}
\end{figure}

\section{Conclusion}
In this paper, we used information theory methods to study the flow of information in DNNs. We determined how entropy and conditional entropy of information changes between consecutive layers and using the Information Bottleneck principle we modeled the learning process of a feed-forward neural network as a constrained optimization problem. Furthermore, we proved an upper bound for the total compression rate of information that can be achieved in a neural network while the overall distortion in the output layer with respect to a desired output is in an acceptable range. In this paper, we assumed that the neural network is used for supervised learning tasks and the input/output spaces are based on discrete alphabets. For the future work, we aim to extend our work to a broader range of learning problems and to include continues input/output spaces in our model.

\bibliographystyle{IEEEbib}
\bibliography{itw}

\begin{thebibliography}{10}

\bibitem{chen14}
X.W. Chen and X.Lin,
\newblock ``Big data deep learning: Challenges and perspectives,''
\newblock {\em IEEE Access}, vol. 2, no. 2, pp. 514--525, 1991.

\bibitem{sun14}
Y.~Sun, X.~Wang, and X.~Tang,
\newblock ``Deep learning face representation from predicting 10,000 classes,''
\newblock in {\em Proc. of the CVPR'14}, 2014, pp. 1891--1898.

\bibitem{lenz15}
I.~Lenz, H.~Lee, and A.~Saxena,
\newblock ``Deep learning for detecting robotic grasps,''
\newblock {\em The International Journal of Robotics Research}, vol. 34, no.
  4-5, April 2015.

\bibitem{tom97}
T.M. Mitchell,
\newblock {\em Machine Learning},
\newblock McGraw-Hill, 1997.

\bibitem{nntr14}
J.~Schmidhuber,
\newblock ``Deep learning in neural networks: An overview,''
\newblock Tech. {R}ep. IDSIA-03-14, The Swiss AI Lab IDSIA, University of
  Lugano \& SUPSI, Swiss, 2014.

\bibitem{hinton06}
G.~E. Hinton and R.~R. Salakhutdinov,
\newblock ``Reducing the dimensionality of data with neural networks,''
\newblock {\em Science}, vol. 313, no. 5786, pp. 504--507, 2006.

\bibitem{erhan13}
C.~Szegedy, A.~Toshev, and D.~Erhan,
\newblock ``Deep neural networks for object detection,''
\newblock in {\em Proc. of NIPS'13}, 2013, pp. 2553--2561.

\bibitem{tishby15}
N.Tishby and N.Zaslavsky,
\newblock ``Deep learning and the information bottleneck principle,''
\newblock in {\em Proc. of ITW'15}, 2015.

\bibitem{bishop96}
C.~M. Bishop,
\newblock ``Theoretical foundations of neural networks,''
\newblock in {\em Proceedings of Physics Computing}, 1996, pp. 500--507.

\bibitem{mehta14}
P.~{Mehta} and D.~J. {Schwab},
\newblock ``{An exact mapping between the Variational Renormalization Group and
  Deep Learning},''
\newblock {\em ArXiv e-prints}, Oct. 2014.

\bibitem{tishby99}
N.~{Tishby}, F.~C. {Pereira}, and W.~{Bialek},
\newblock ``The information bottleneck method,''
\newblock in {\em Proceedings of 37th Annual Allerton Conference on
  Communication, Control and Computing}, 1999.

\bibitem{dietterich2002}
T.G. Dietterich,
\newblock {\em Structural, Syntactic, and Statistical Pattern Recognition:
  Joint IAPR International Workshops SSPR 2002 and SPR 2002 Windsor, Ontario,
  Canada, August 6--9, 2002 Proceedings}, chapter Machine Learning for
  Sequential Data: A Review, pp. 15--30,
\newblock Springer Berlin Heidelberg, Berlin, Heidelberg, 2002.

\bibitem{blahut72}
R.E. Blahut,
\newblock ``Computation of channel capacity and rate-distortion functions,''
\newblock {\em Information Theory, IEEE Trans. on}, vol. 18, no. 4, pp.
  460--473, April 1972.

\end{thebibliography}

\end{document}